\documentclass[draft]{article}
 
\usepackage[bottom=1.3in, top=1.3in, left=1.3in,right=1.3in]{geometry}

\usepackage{amsmath,amssymb, amsthm, mathrsfs} 
\usepackage{xspace}
\usepackage{color}
\usepackage{subfig}
\usepackage{hyperref}
\usepackage{url}

\usepackage{thmtools}
\usepackage{thm-restate}

\usepackage{tikz}
\usetikzlibrary{positioning}
\usetikzlibrary{shapes.geometric}
\tikzstyle{Variable} = [node distance=1em,inner sep=0pt,minimum size=.5em]
\tikzstyle{Grey}=[draw,thick,circle,color=black,fill=black!20!white]
\tikzstyle{Clause} = [draw, circle,fill,node distance=3em,inner sep=0pt,minimum size=.5em]
\tikzstyle{Edge} = [fill,color=black, very thick]
\tikzstyle{vertex} = [fill,shape=circle,node distance=2.5em,inner sep=0pt,minimum size=.5em]
\tikzstyle{dvertex} = [draw,thick,color=black, fill=white,circle,node distance=2.5em,inner sep=0pt,minimum size=.5em]
\tikzstyle{edge} = [fill,color=black,opacity=.1,fill opacity=.5,line cap=round, line join=round, line width=2em]

\usepackage{palatino}

\usepackage{enumitem}
\setlist[description]{font=\normalfont\scshape}

\newtheorem{theorem}{Theorem}[section]
\newtheorem{definition}[theorem]{Definition}
\newtheorem{lemma}[theorem]{Lemma}
\newtheorem{corollary}[theorem]{Corollary}

\newcommand{\retheorem}{theorem}

\newcommand{\restatetheorem}[1]{\renewcommand{\retheorem}{rtheorem}#1\renewcommand{\retheorem}{theorem}}

\renewcommand{\phi}{\varphi}

\newcommand{\calD}{\ensuremath{\mathcal{D}}\xspace}
\newcommand{\calH}{\ensuremath{\mathcal{H}}\xspace}

\newcommand{\matching}[1][2,4]{$(#1)$-matching\xspace}
\newcommand{\matchings}[1][2,4]{$(#1)$-matchings\xspace}

\newcommand{\memory}[1]{\ensuremath{\mathfrak{M}_{#1}}\xspace}

\newcommand{\FPF}{flippable product-family\xspace}
\newcommand{\FPFs}{flippable product-families\xspace}

\newcommand{\WS}[1]{$#1$-winning strategy\xspace}
\newcommand{\WSs}[1]{$#1$-winning strategies\xspace}

\newcommand{\PCR}  {\ensuremath{\mathsf{PCR}}\xspace}
\newcommand{\RES}  {\ensuremath{\mathsf{RES}}\xspace}

\DeclareMathOperator{\Space}{MSpace}
\DeclareMathOperator{\dom}{dom}
\newcommand{\st}{\ :\ }

\newcommand{\Star}{\ensuremath{\lambda}\xspace}

\newcommand{\HH}{\ensuremath{\mathcal{H}}\xspace}
\newcommand{\R}{\ensuremath{\mathcal{R}}\xspace}

\newcommand{\xvec}[3]{\ensuremath{{#1}_{#2},\ldots,{#1}_{#3}}\xspace}

\newcommand{\CoverGame}[2][2,4]{%
\ensuremath{\mathsf{CoverGame}_{(#1)}(#2)}\xspace%
}

\newcommand{\LL}{\ensuremath{\mathscr{L}}\xspace}
\newcommand{\FF}{\ensuremath{\mathscr{F}}\xspace}

\newcommand{\Choose}{{\sf Choose}\xspace}
\newcommand{\Cover}{{\sf Cover}\xspace}
\newcommand{\CovGam}{{\sf CoverGame}\xspace}

\newcommand{\robust}{robust\xspace}
\newcommand{\expander}[1]{$(#1)$-bipartite expander\xspace}

\newcommand{\dfn}[1]{\textbf{\textit{#1}}\xspace}

\begin{document}
\title{Space proof complexity for random $3$-CNFs \\ via a $(2-\epsilon)$-Hall's Theorem}
\author{
Ilario Bonacina%
\footnote{
Computer Science Department, Sapienza University of Rome, via Salaria 113, 00198 Rome, Italy, \texttt{\{bonacina, galesi, huynh, wollan\}@di.uniroma1.it}.
}
\and 
Nicola Galesi$^*$
\and 
Tony Huynh$^*$%
\footnote{Supported by the European Research Council under the European Union's Seventh Framework Programme (FP7/2007-2013)/ERC Grant Agreement no. 279558.}
\and 
Paul Wollan$^{* \dagger}$
}
\date{\today}
\maketitle

\begin{abstract}
We investigate the space complexity of refuting  $3$-CNFs in Resolution and  algebraic systems.  No lower bound for refuting any family of $3$-CNFs was previously known for the {\em total space} in resolution or for the {\em monomial space} in algebraic systems.  
Using the framework of \cite{bg15}, we prove  that every {\em Polynomial Calculus with Resolution} refutation of a random $3$-CNF $\phi$ in $n$ variables requires, with high probability, $\Omega(n/\log n)$  {\em distinct monomials} to be kept simultaneously in memory.  
The same construction also proves that every {\em Resolution} refutation $\phi$ requires, with high probability, $\Omega(n/\log n)$ clauses each of width $\Omega(n/\log n)$ to be kept at the same time in memory. This gives a $\Omega(n^2/\log^2 n)$ lower bound for the {\em total space} needed in Resolution to refute $\phi$. The results answer questions about space complexity of $3$-CNFs posed in \cite{FilmusLNTR12,FilmusLMNV13,BonacinaGT14,bg15}.

The main technical innovation is a variant of  {\em Hall's theorem}.   
We show that in bipartite graphs $G$  with bipartition $(L,R)$ and left-degree at most 3, $L$ can be covered by certain families of disjoint paths, called {\em \matchings}, provided that $L$ {\em expands} in $R$ by a factor of $(2-\epsilon)$, for $\epsilon < \frac{1}{23}$. 
\end{abstract}

\section{Introduction}
During the last decade, an active line of research in proof
complexity has been the space complexity of proofs and how space
is related to other complexity measures (like size, length, width, degree) \cite{EstebanT01, AlekhnovichBRW02,Ben-SassonG03,Ben-Sasson02,AtseriasD08,Ben-SassonN08,Nordstrom09, BenSassonN11,FilmusLNTR12,FilmusLMNV13,BonacinaGT14,bg15}. 
This investigation has raised several important foundational questions.  Some of these have been solved, while several others are still open and challenging (see \cite{Nordstrom13} for a survey on this topic). 
Space of proofs concerns the minimal memory occupation of algorithms verifying the correctness of proofs in concrete propositional proof systems, and is thus also relevant in more applied algorithmic contexts.  
For instance, a major problem in state of the art SAT-solvers is memory consumption. In proof complexity, this resource is modeled by proof space.  
It is well-known that SAT-solvers used in practice (like CDCL) are based on low-level proof systems such as {\em Resolution}. 

In this work we focus on two well known proof systems that play a central role in proof complexity: {\em Resolution} \cite{Robinson:1965} and {\em Polynomial Calculus} \cite{CleggEI96}. 
Resolution (\RES) is a refutational proof system for unsatisfiable propositional CNF formulas using only one logical rule: $ \frac{A \vee x \;\;\;\;\; \neg x \vee B}{A\vee B}$.  
Polynomial calculus is an algebraic refutational proof system for unsatisfiable sets of polynomials (over $\{0,1\}$ solutions) based on two rules: {\em linear combination} of polynomials and {\em multiplication} by variables.  In this article, we consider the stronger system {\em Polynomial Calculus with Resolution} (\PCR) which extends both Resolution and Polynomial Calculus \cite{AlekhnovichBRW02}.

Several different measures for proof space  were investigated for these two systems \cite{EstebanT01,AlekhnovichBRW02,Ben-Sasson02,AtseriasD08,Ben-SassonN08,Nordstrom09, BenSassonN11,FilmusLNTR12,BonacinaGT14,bg15}.  
In this work we focus on {\em total space} (for \RES), which is the maximum number of variables (possibly repeated) to be kept simultaneously in  memory while verifying a proof; and {\em monomial space} (for \PCR), which is the maximum number of distinct monomials to be kept simultaneously  in  memory while verifying a proof.  
Both measures were introduced in  \cite{AlekhnovichBRW02}, where some preliminary lower and upper bounds were given. 
However, major open problems about these two measures were solved only recently in \cite{FilmusLMNV13, bg15,BonacinaGT14}.  
In particular, in \cite{bg15,BonacinaGT14} it was proved that, for $r\geq 4$, random $r$-CNFs over $n$ variables require $\Theta(n^2)$ total space in resolution and $\Theta(n)$ monomial space in \PCR. 

It is an open problem whether there is any family of $3$-CNFs requiring large {\em total space} (in \RES) and {\em monomial space} (in \PCR).
In this work we solve this problem proving both {\em total space} (in \RES) and {\em monomial space} (in \PCR) lower bounds for the case of random $3$-CNFs, a computationally hard class of formulas that are conjectured difficult to prove in very strong proof systems.  

\paragraph{Results.}

Let $\phi$ be a random $3$-CNF in $n$ variables.  We prove that every \PCR  refutation of $\phi$  requires, with high probability, $\Omega(\frac{n}{\log n})$  distinct monomials to be kept simultaneously in memory (Corollary \ref{cor:rand}).  
Moreover, every \RES refutation of $\phi$ has, with high probability, $\Omega(\frac{n}{\log n})$ clauses each of width $\Omega(\frac{n}{\log n})$ to be kept at the same time in memory (Corollary \ref{cor:rand}). This gives a $\Omega(n^2/\log^2 n)$ lower bound for the total space of every \RES refutation of $\phi$.  These results resolve questions about space complexity of $3$-CNFs mentioned in \cite{FilmusLMNV13,BonacinaGT14,bg15,FilmusLNTR12}.

Both results follow using the framework proposed in \cite{bg15}, where the construction of suitable families of assignments called {\em \WSs{k}}  (Definition \ref{def:kex}) leads to monomial space lower bounds in \PCR (Theorem \ref{thm:lowerbound}).
This construction is made possible by a modification of Hall's theorem \cite{Hall} for matchings to \matchings. 
A {\em \matching} in a bipartite graph with bipartition $(L,R)$ is a collection of disjoint paths of length at most $4$ with both endpoints in $R$. 
The notion of \matching[h,k] (Definition \ref{def:h-k-matchings}) is a generalization of the standard notion of matching. With this terminology a matching is a \matching[1,1] (see Figure 1.(a)), and a {\em 2-matching} (as it is used in \cite{BonacinaGT14,bg15}) is a \matching[2,2] (see Figure 1.(b)). Figure 1.(c) depicts a 
\matching.
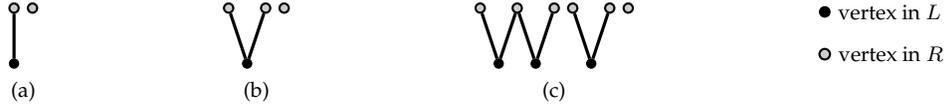
\begin{figure}
\scriptsize
\centering
\hfill
\subfloat[]{
\begin{tikzpicture}
\node[Variable, Grey] (v10) {};
\node[Variable, Grey, right of = v10] (v11) {};
\node[Clause, below of=v10] (v12) {};
\draw[Edge] (v12) -- (v10);
\end{tikzpicture}
}
\hfill
\subfloat[]{
\begin{tikzpicture}
\node[Variable,  Grey] (v8) {};
\node[Variable,  right of = v8] (v9) {};
\node[Variable, Grey, right of = v9] (v10) {};
\node[Variable, Grey, right of = v10] (v11) {};
\node[Clause, below of=v9] (v12) {};
\draw[Edge] (v8) -- (v12) -- (v10);
\end{tikzpicture}
}
\hfill
\subfloat[]{
\begin{tikzpicture}
\node[Variable,  Grey] (v1) {};
\node[Variable, right of = v1] (v2) {};
\node[Variable, Grey, right of = v2] (v3) {};
\node[Clause, below of= v2] (v4) {};
\node[Variable, right of = v3] (v5) {};
\node[Variable, Grey,right of = v5] (v6) {};
\node[Clause, below of= v5] (v7) {};
\node[Variable,  Grey,right of = v6] (v8) {};
\node[Variable,  right of = v8] (v9) {};
\node[Variable, Grey, right of = v9] (v10) {};
\node[Variable, Grey, right of = v10] (v11) {};
\node[Clause, below of=v9] (v12) {};
\draw[Edge] (v1) -- (v4) -- (v3) -- (v7) -- (v6);
\draw[Edge] (v8) -- (v12) -- (v10);
\end{tikzpicture}
}
\hfill
\subfloat{
\footnotesize
\begin{tikzpicture}
\node[Clause,node distance=2em,label=right:vertex in $L$ ](v1){};
\node[Variable, Grey,node distance=2em,label=right:vertex in $R$, below of=v1](v2){};
\end{tikzpicture}
}
\caption{Comparison among (a) \matchings[1,1], (b) \matchings[2,2], (c) \matchings[2,4]}
\label{fig:matching}
\end{figure}

We can now state our variant of Hall's theorem.   

\begin{restatable}[$(2-\epsilon)$-Hall's Theorem]{\retheorem}{restateHallone}
\label{thm:Hall1}
Let $ \epsilon < \frac{1}{23}$.  Let $G$ be a bipartite graph with bipartition $(L,R)$ such that the left-degree is at most $3$ and for every subset $A$ of $L$, $|N_G(A)|\geq (2-\epsilon)|A|$.  Then there exists a \matching covering $L$.
\end{restatable}

This theorem and its proof are independent from the proof complexity results and might be usable in other contexts. 
Note that the converse of Theorem \ref{thm:Hall1} does not hold (unlike in Hall's theorem).  Hence, in order to apply it as a substitute for Hall's theorem, we had to prove a slightly stronger statement.

\begin{restatable}{\retheorem}{restateHall}\label{thm:Hall2}
Let $ \epsilon < \frac{1}{23}$.  Let $G$ be a bipartite graph with bipartition $(L,R)$ such that each vertex in $L$ has degree at most $3$ and no pair of degree $3$ vertices in $L$ have the same set of neighbours. If $|N_G(L)|\geq (2-\epsilon)|L|$, and each proper subset of $L$ can be covered by a \matching, then $L$ can be covered by a \matching. 
\end{restatable}

\paragraph{Outline of the paper.}

Section \ref{sec:prelimin} contains some preliminary notions about proof complexity. In particular, we present the formal definitions of Resolution and Polynomial Calculus with Resolution, the model of space (based on \cite{EstebanT01, AlekhnovichBRW02}) and the formal definition of total space and monomial space.

In Section \ref{sec:hall}, we define the notion of \matchings[h,k] (Definition \ref{def:h-k-matchings}) and present the proofs of Theorems \ref{thm:Hall1} and \ref{thm:Hall2} guaranteeing the existence of \matchings.  The proof of the main theorem of Section \ref{sec:hall} relies on a concentration result on the average right-degree and a discharging argument.  We also prove a bound for the best possible value  of $\epsilon$ for which Theorems \ref{thm:Hall1} and \ref{thm:Hall2} could hold and conjecture that this bound is in fact the optimal value of $\epsilon$.  

In Section \ref{sec:covergame}, we define a two player covering game \CovGam, whose aim is to dynamically build a \matching  inside a fixed bipartite graph $G$ (Definition \ref{thm:covergame}). 
Informally, a player, \Choose, queries nodes in the graph $G$ and the other player, \Cover, attempts to extend the current \matching to also cover the node queried (if not already covered).  
The main result of Section \ref{sec:covergame} is Theorem \ref{thm:covergame}, where we prove that if the graph $G$ has large left-expansion (i.e. large enough to apply Theorem \ref{thm:Hall2} to sufficiently large subgraphs of $G$), then there is a winning strategy for \Cover to force \Choose to query a very large portion of the graph $G$.
In the analysis of the game, we use the $(2-\epsilon)$-Hall's Theorem and  \matchings in a similar manner to how \matchings[1,1] and \matchings[2,2] were used in \RES and \PCR\cite{Ben-SassonG03,Atserias04,BonacinaGT14,bg15}. 

In Section \ref{sec:lowerbounds}, we present a simplified (but less general) version of the \WSs{k} of \cite{bg15} (Definition \ref{def:kex}). 
These \WSs{k} were used in \cite{bg15} to prove monomial space lower bounds for \PCR. 
We recall the definition of {\em  $r$-free family} (Definition \ref{def:r-free}), used in \cite{BonacinaGT14} to prove total space lower bounds in \RES (Theorem \ref{thm:main-free}). 
We then show that the existence of a \WS{k} for a formula $\phi$ implies the existence of $(k-1)$-free families of assignments for the formula $\phi$ (Lemma \ref{lem:winning-free})\footnote{Notice that this lemma is not proving a reduction of the notion of $r$-free families to the full general version of \WSs{k} in \cite{bg15}.}.
We prove (Theorem \ref{thm:cover-to-strategy}) that if \Cover wins \textsf{CoverGame} on the adjacency graph of a CNF $\phi$  (see Section 2 for the definition) guaranteeing \matchings of maximal size $\mu$, then there exists a \WS{\mu} for the polynomial encoding of $\phi$. 
Finally, the monomial space in \PCR and the total space in \RES for random $3$-CNFs (Corollary \ref{cor:rand}) follow from well-known results about expansion of its adjacency graph \cite{ChvatalS88,BeameP96,Ben-SassonW01,Ben-SassonG03}.

\section{Preliminaries}\label{sec:prelimin}

Let $X$ be a set of variables. A \dfn{literal} is a boolean constant, $0$ or $1$, or a variable $x\in X$, or the negation $\neg x$ of a variable $x$. 
A \dfn{clause} is a disjunction of literals: $C = (\ell_1 \vee \ldots \vee \ell_k)$. The \dfn{width} of a clause is the number of literals in it. 
A formula $\phi$ is in \dfn{Conjunctive Normal Form} (CNF) if $\phi = C_1 \wedge \ldots  \wedge C_m$ where $C_i$ are clauses. It is a $k$-CNF if each $C_i$ contains at most $k$ literals.
Let $\phi$ be a CNF and $X$ be the set of variables appearing in $\phi$. 
The \dfn{adjacency graph of $\phi$} is a bipartite graph $G_\phi$ with bipartition $(L,R)$ such that $L$ is the set of clauses of $\phi$,  $R=X$, and $(C,x)\in E$ if and only if $x$ or $\neg x$ appears in $C$.  If $\phi$ is a $k$-CNF, then $G_\phi$ has left-degree at most $k$.

We define $\overline{X}= \{\bar{x}\ :\ x \in X\}$, which we regard as a set of formal variables with the intended meaning of $\bar x$ as $\neg x$. 
Given a field $\mathbb{F}$,  the ring $\mathbb{F}[X,\overline{X}]$ is the ring of polynomials in the variables $X\cup \overline{X}$ with coefficients in $\mathbb{F}$. 
Following \cite{AlekhnovichBRW02}, we use the following \dfn{standard encoding ($tr$)} of CNF formulas over $X$ into a set of polynomials in $\mathbb{F}[X,\overline X] $: 
$tr(\phi)= \{tr(C)\ :\ C\in \phi\}\cup\{x^2-x, x+\bar x-1 \st x\in X\}$, where
$$
tr(x) = \bar x, \qquad  tr(\neg x) = x ,\qquad  tr( \bigvee_{i=1}^{n} \ell_i ) = \prod_{i=1}^n tr(\ell_i).
$$
A set of polynomials $P$ in $\mathbb{F}[X]$ is \dfn{contradictory} if and only if  $1$ is in the ideal generated by $P$. 
Notice that a CNF $\phi$ is unsatisfiable if and only if $tr(\phi)$ is a contradictory set of polynomials.

\dfn{Resolution} (\RES) \cite{Robinson:1965} is a propositional proof system for refuting unsatisfiable CNFs. 
Starting from an unsatisfiable CNF $\phi$, \RES allows us to derive the empty clause $\bot$ using the following inference rule:
$$
\frac{C\vee x \quad D\vee \neg x}{C\vee D}.
$$

In order to study space of proofs we follow a model inspired by the definition of space complexity for Turing machines, where a machine is given a read-only input tape from which it can download parts of the input to the working memory as needed \cite{EstebanT01}. 

Given an unsatisfiable CNF formula $\phi$, a \dfn{\RES refutation} of $\phi$ is a sequence  $\Pi=\langle\memory{0},\ldots,\memory{\ell}\rangle$ of sets of clauses, called \dfn{memory configurations}, such that: $\memory{0} =\emptyset$,  $\bot\in \memory{\ell}$, and for all $i\leq \ell$, $\memory{i}$ is obtained by $\memory{i-1}$ by applying one of the following rules:

\begin{quote}
\textsc{(Axiom Download)} $\memory{i} = \memory{i-1} \cup \{C\}$, where $C$ is a clause of $\phi$;\\
\textsc{(Inference Adding)} $\memory{i} = \memory{i-1} \cup \{C\vee D\}$, where $C\vee x,D\vee \neg x\in \memory{i-1}$, for some variable $x$;\\
\textsc{(Erasure)} $\memory{i} \subset \memory{i-1}$.
\end{quote}

The \dfn{total space} of $\Pi$ is the maximum over $i$ of the number of literals (counted with repetitions) occurring in \memory{i}.

\dfn{Polynomial Calculus with Resolution} (\PCR) \cite{AlekhnovichBRW02} is an algebraic proof system for polynomials in $\mathbb{F}[X,\overline{X}]$. 
Starting from an initial set of contradictory polynomials $P$  in $\mathbb{F}[X,\overline{X}]$, \PCR allows us to derive the polynomial $1$ using the following inference rules: for all $p, q\in \mathbb{F}[X,\overline{X}]$
$$
\frac{p \quad \quad q }{\alpha p+\beta q}\ \forall\alpha,\beta \in \mathbb{F}, 
\qquad\qquad\qquad 
\frac{\quad p \quad}{vp} \forall v \in X\cup \overline{X}.
$$

To force 0/1 solutions, we always include the \dfn{boolean axioms} $\{x^2-x, x +\overline{x}-1\}_{x \in X}$ among the initial polynomials, as in the case of the polynomial encoding of CNFs. 

Given a set of contradictory polynomials $P$, a \dfn{\PCR refutation} of $P$ is a sequence  $\langle\memory{0},\ldots,\memory{\ell}\rangle$ of sets of polynomials, called \dfn{memory configurations}, such that: $\memory{0} =\emptyset$,  $1\in \memory{\ell}$, and for all $i\leq \ell$, $\memory{i}$ is obtained by $\memory{i-1}$ by applying one of the following rules:
\begin{quote}
\textsc{(Axiom Download)} $\memory{i} = \memory{i-1} \cup \{p\}$, where $p\in P$;\\
\textsc{(Inference Adding)} $\memory{i} = \memory{i-1} \cup \{p\}$, where $p$ is some polynomial inferred from polynomials occurring in $\memory{i-1}$ using the inference rules of \PCR;\\
\textsc{(Erasure)} $\memory{i} \subset \memory{i-1}$.
\end{quote}

The \dfn{monomial space} of $\Pi$, denoted by $\Space(\Pi)$, is the maximum over $i$ of the number of distinct monomials appearing in \memory{i}.

If in the definition of \PCR refutation we substitute the \textsc{Inference Adding} rule with:
\begin{quote}
 \textsc{(Semantical Inference)} $\memory{i}$ is contained in the  ideal generated by $\memory{i-1}$,
\end{quote}
we have what is called a \dfn{semantical \PCR refutation} of  $P$ \cite{AlekhnovichBRW02}.

\subsection{Partial assignments}

A \dfn{partial assignment} over a set of variables $X$ is a map $\alpha:X\longrightarrow \{0,1,\star\}$. The \dfn{domain} of $\alpha$ is $\dom(\alpha)=\alpha^{-1}(\{0,1\})$. We denote the partial assignment with empty domain as \Star.
Given a partial assignment $\alpha$ and a CNF $\phi$ we can apply $\alpha$ to $\phi$, obtaining a new formula $\alpha(\phi)$ in the standard way, substituting each variable $x$ of $\phi$ in $\dom(\alpha)$ with the value $\alpha(x)$ and then simplifying the result. 
We say that $\alpha$ {\em satisfies} $\phi$, and we write  $\alpha\models \phi$, if $\alpha(\phi)=1$. 
Similarly, for a family $F$ of partial assignments, $F\models \phi$ means that for each $\alpha\in F$, $\alpha\models \phi$.

For each partial assignment $\alpha$ over $X \cup \overline{X}$ we assume that it respects the intended meaning of the variables; that is, $\alpha(\bar x)=1-\alpha(x)$ for each $x,\bar x \in \dom(\alpha)$. In particular, it is always possible to extend to $X \cup \overline{X}$ an assignment $\beta$ over $X$  respecting the previous property.
Given an assignment $\alpha$ and a polynomial $p$ in $\mathbb{F}[X,\overline X]$, we can apply $\alpha$ to $p$, obtaining a new polynomial $\alpha(p)$ in the standard way,
substituting each variable $x$ of $p$ in $\dom(\alpha)$ with the value $\alpha(x)$ and then simplifying the result.
The notation $\alpha\models p$ means that $\alpha(p)=0$.
If $F$ is a family of partial assignments and $P$ a set of polynomials, we write $F\models P$ if $\alpha\models p$ for each $\alpha\in F$ and $p\in P$. 
Notice that if $\phi$ is a CNF and $\alpha$ is a partial assignment then $\alpha\models\phi$ if and only if $\alpha(tr(\phi))=0$. 

\section{A $\boldsymbol{(2-\epsilon)}$-Hall's Theorem}\label{sec:hall}

We begin with the definition of \matching[h,k].  
\begin{definition}[{\matching[h,k]}]\label{def:h-k-matchings}
Let $k\geq h$ be positive integers and $G$ a bipartite graph with bipartition $(L,R)$. An \dfn{\matching[\boldsymbol{h,k}] in $G$} is a subgraph $F$ of $G$ such that
\begin{enumerate}
\item each connected component of $F$ is a tree with at most $k$ edges;
\item each vertex in $V(F)\cap L$ has degree exactly $h$ in $F$.
\end{enumerate}
An \matching[h,k] $F$ {\em covers} a set of vertices $S$ if $S\subseteq V(F)$. Define $\ell(F)=V(F)\cap L$ and $r(F)=V(F)\cap R$.
\end{definition}

We now proceed with the proof of Theorem \ref{thm:Hall2}; we will see afterwards that Theorem \ref{thm:Hall1} follows easily. 

\restatetheorem{\restateHall*}

\begin{proof}
Define the hypergraph $\calH=(V, E)$ with $V=N_G(L)$ and $E=\{N_G(x): x \in L\}$.  
Observe that each vertex $v$ in $L$ has degree 2 or 3, otherwise $v$ could not be covered by a $(2,4)$-matching.  Similarly, no degree 2 vertices in $L$ have the same neighbourhood.  By assumption, no degree 3 vertices have the same neighbourhood.  Therefore, $N_G: L \rightarrow E$ is a bijection and $|V|\geq (2-\epsilon)|L|=(2-\epsilon)|E|$.  Recall that the degree of a vertex $v$ in $\calH$ is the number of distinct edges which contain $v$.

Let $L' \subseteq L$ and let $E' = \{N_G(x): x \in L'\}$.  The existence of a \matching in $G$ covering $L'\subseteq L$ is equivalent to the existence of an injective function $f:E' \rightarrow \{\{x, y\}: x, y \in N_G(L')\}$ such that for every $e \in E'$,  $f(e)\subseteq \binom{e}{2}$ and for each triple of distinct hyperedges $e_1,e_2,e_3\in E'$, it is not the case that $f(e_1)$ intersects $f(e_2)$ and $f(e_2)$ intersects $f(e_3)$.
We call a function $f$ with these properties a {\em $2$-path cover} of $E'$ in \calH.

Observe that all the configurations of hyperedges shown in Figure \ref{fig:reductions} have a 2-path cover using only degree $1$ and $2$ vertices of \calH.
If any of these configurations appear in \calH, we can by assumption find a 2-path cover $f$ of the remaining hyperedges, and then extend $f$ to a 2-path cover of \calH.
Therefore, we may assume that no configuration from Figure \ref{fig:reductions} appears in \calH.

\begin{figure}[!ht]
\vspace{-1.5em}
\scriptsize
\centering
\pgfdeclarelayer{background}
\pgfsetlayers{background,main}
\hfill
\subfloat[]{
\begin{tikzpicture}
\node[dvertex] (v1) {};
\node[vertex,below of=v1] (v2) {};
\node[dvertex,right of=v2] (v3) {};

\draw [red,very thick,dashed] (v2) circle (.5em);

\begin{pgfonlayer}{background}
\draw[edge] (v1) -- (v2) -- (v3)--(v1);
\end{pgfonlayer}
\end{tikzpicture}
}
\hfill
\subfloat[]{
\begin{tikzpicture}
\node[dvertex] (v4) {};
\node[dvertex, right of=v4] (v5) {};

\begin{pgfonlayer}{background}
\draw[edge] (v4) -- (v5) --(v4);
\end{pgfonlayer}
\end{tikzpicture}
}
\hfill
\subfloat[]{
\begin{tikzpicture}
\node[dvertex] (v1) {};
\node[vertex,below of=v1] (v2) {};
\node[vertex,right of=v2] (v3) {};
\node[vertex,right of=v3] (v4) {};
\node[dvertex,above of=v4] (v5) {};

\draw [red,very thick,dashed] (v2) circle (.5em);
\draw [red,very thick,dashed] (v4) circle (.5em);

\begin{pgfonlayer}{background}
\draw[edge] (v1) -- (v2) -- (v3)--(v1);
\draw[edge] (v3) -- (v4) -- (v5) -- (v3);
\end{pgfonlayer}
\end{tikzpicture}
}
\hfill
\subfloat[]{
\begin{tikzpicture}
\node[vertex] (v6) {};
\node[dvertex,above of=v6] (v7) {};
\node[dvertex,right of=v6] (v8) {};

\begin{pgfonlayer}{background}
\draw[edge] (v6) -- (v7) -- (v6);
\draw[edge] (v6) -- (v8) --(v6);
\end{pgfonlayer}
\end{tikzpicture}
}
\hfill
\subfloat[]{
\begin{tikzpicture}
\node[vertex] (v9) {};
\node[dvertex, above of=v9] (v10) {};
\node[vertex, right of=v9] (v11) {};
\node[dvertex, right of=v11] (v12) {};

\draw [red,very thick,dashed] (v9) circle (.5em);

\begin{pgfonlayer}{background}
\draw[edge] (v9) -- (v10) -- (v11)--(v9);
\draw[edge] (v11) -- (v12) -- (v11);
\end{pgfonlayer}
\end{tikzpicture}
}
\hfill
\subfloat{
\begin{tikzpicture}
\footnotesize
\node[dvertex,node distance=2em,label=right:vertex of degree $1$ in $\calH$](v1){};
\node[vertex,node distance=2em,label=right:vertex of degree $2$ in $\calH$, below of=v1](v2){};
\node[vertex,node distance=2em,label=right:\text{ vertex of degree 1,2 or more in $\calH$}, below of=v2](v3){};
\draw [red,very thick,dashed] (v3) circle (.5em);
\end{tikzpicture}
}
\caption{A set of reducible configurations for $\mathcal{H}$}
\label{fig:reductions}
\end{figure}
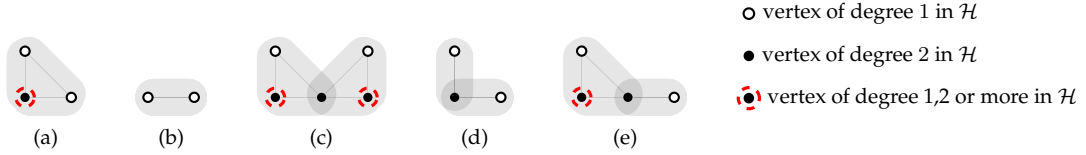

Let $d$ be the average degree in \calH. Observe that
\[
3|E|\geq \sum_{v \in V} \deg_{\calH}(v)=d |V| \geq d(2-\epsilon)|E|,
\]
where the last inequality follows from the hypothesis that $|V|\geq (2-\epsilon)|E|$.  Thus, $d \leq \frac{3}{2-\epsilon}$.

Let $D$ be the set of vertices of \calH of degree 1, and \calD be the set of hyperedges which contain a vertex in $D$.  
Note that $|D|=|\calD|$, since the configurations $(a)$, $(b)$ of Figure \ref{fig:reductions} do not appear in \calH.  We now prove a concentration result for $|D|$.
An upper bound follows immediately from the following chain of inequalities
$$
|D|=|\calD|\leq |E|\leq \frac{1}{2-\epsilon}|V|.
$$

For the lower bound suppose $\frac{1-2\epsilon}{2-\epsilon}|V| > |D|$.  Then,
\begin{align*}
\frac{3}{2-\epsilon} |V|  \geq d|V| = \sum_{v\in D }\deg_\calH(v) +\sum_{v\in V\setminus D}\deg_\calH(v) \\
\geq |D|+2|V\setminus D|> \frac{1-2\epsilon}{2-\epsilon} |V| + 2 (1-\frac{1-2\epsilon}{2-\epsilon}) |V|=\frac{3}{2-\epsilon} |V|,
\end{align*}
which is a contradiction. Hence we have that $\frac{1-2\epsilon}{2-\epsilon}|V| \leq |D| \leq \frac{1}{2-\epsilon}|V|$.

We finish the proof with a discharging argument.  Each vertex $v$ of \calH receives a charge of $\deg_{\calH}(v)$.  Let $E_2$ be the set of hyperedges in $E$ of size $2$ and $\calD_2$ be $\calD\cap E_2$. Clearly $|\calD_2|\leq |E_2|$. The edges in 
$\calD_2$ receive a charge of -2, the edges in $\calD\setminus \calD_2$ receive charge $-3$.  The edges not in \calD receive no charge.

We now perform the following discharging rule.  Each
hyperedge $e$ in \calD gives a charge of -1 to each vertex in $e$.  After discharging, every hyperedge has charge 0, and every vertex has non-negative charge.
Let $Z$ be the set of vertices with charge 0 after discharging.  Observe that a vertex $x$ is in $Z$ if and only if every hyperedge containing $x$ also contains a degree 1 vertex.
 
Let $C$ denote the total charge.  Then,
\begin{align*}
C= 3|E|-|E_2|-3|\calD|+|\calD_2|\leq 3|E|-3|D| \leq \frac{3}{2-\epsilon} |V|-3\frac{1-2\epsilon}{2-\epsilon} |V|=\frac{6\epsilon}{2-\epsilon}|V|.
\end{align*}
It follows that $|Z| \geq \frac{2-7\epsilon}{2-\epsilon} |V|$.  Since $D \subseteq Z$ and $|D| \leq \frac{|V|}{2-\epsilon}$, we conclude that $|Z \setminus D| \geq \frac{1-7\epsilon}{2-\epsilon} |V|$.
Because  the  configurations $(c),(d),(e)$ of Figure \ref{fig:reductions} do not appear in \calH, every vertex in $Z \setminus D$ has degree at least 3.  Thus,
\begin{align*}
\frac{3}{2-\epsilon} |V|   &\geq d|V| \geq \sum_{v \in D} \deg_{\calH} (v)+ \sum_{v \in Z \setminus D} \deg_{\calH} (v) \geq |D| + 3|Z \setminus D| \\&\geq \frac{1-2\epsilon}{2-\epsilon}|V| + 3\frac{1-7\epsilon}{2-\epsilon} |V|  = \frac{4-23\epsilon}{2-\epsilon}|V|.
\end{align*}

Therefore, $3 \geq 4-23\epsilon$, which is a contradiction as $\epsilon < \frac{1}{23}$.
\end{proof}

We now quickly derive Theorem \ref{thm:Hall1}.

\restatetheorem{\restateHallone*}
\begin{proof}
Assume the theorem is false, and pick a counterexample $G$ with $|V(G)|$ minimum.  For every proper subset $L' \subseteq L$, the subgraph of $G$ induced by $L'$ and $N_G(L')$ has a \matching[2,4] covering $L'$ by the minimality of $G$.  But now, Theorem \ref{thm:Hall2} implies that there is a \matching covering $L$, and so $G$ is not a counterexample.   
\end{proof}

We end this section with a comment on the parameter $\epsilon$. 
\begin{theorem}\label{thm:epsilon}
For all $\epsilon>\frac{1}{3}$ there exists a bipartite graph $G_\epsilon$ with bipartition $(L,R)$ such that each vertex in $L$ has degree at most $3$ and no pair of degree $3$ vertices in $L$ have the same set of neighbours. Moreover, $|N_{G_\epsilon}(L)|\geq (2-\epsilon)|L|$ and each proper subset of $L$ can be covered by a \matching but $L$ cannot be covered by a \matching. 
\end{theorem}
We conjecture that Theorem \ref{thm:Hall2} is true for $\epsilon \le \frac{1}{3}$.  Theorem \ref{thm:epsilon} shows that this would be best possible.

We now prove Theorem~\ref{thm:epsilon}, here rephrased in terms of hypergraphs.  

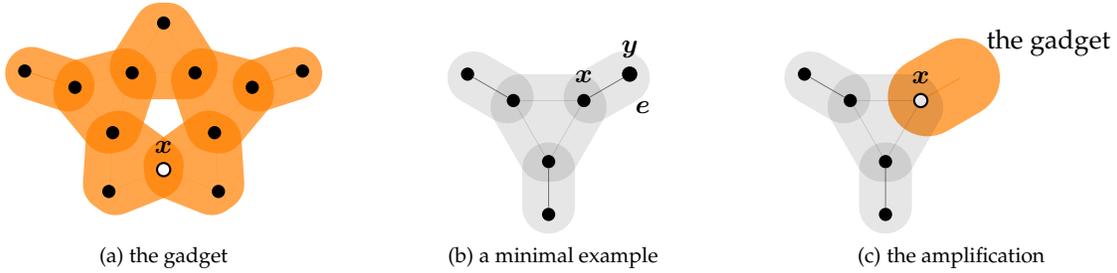
\begin{figure}[hb]
\centering
\pgfdeclarelayer{background}
\pgfsetlayers{background,main}
\subfloat[the gadget]{
\begin{tikzpicture}
\node[draw=none,minimum size=4em,regular polygon,rotate around={180:(0,0)},regular polygon sides=5] (a) {};
\node[draw=none,minimum size=7em,regular polygon,regular polygon sides=5] (b) {};
\node[draw=none,minimum size=11em,regular polygon,regular polygon sides=5] (c) {};
\node[dvertex,label=above:$\boldsymbol{x}$] (v1) at (a.corner 1){};
\foreach \x in {2,...,5}{
  \node[vertex] (v\x) at (a.corner \x) {};
  \node[vertex] (w\x) at (b.corner \x) {};
}
  \node[vertex] (w1) at (b.corner 1) {};
\node[vertex] (z2) at (c.corner 2) {};
\node[vertex] (z5) at (c.corner 5) {};
\begin{pgfonlayer}{background}
\draw[edge, color=orange, opacity=.7] (v1) -- (v2) -- (w4) -- (v1);
\draw[edge, color=orange, opacity=.7] (v2) -- (v3) -- (w5) -- (v2);
\draw[edge, color=orange, opacity=.7] (v3) -- (v4) -- (w1) -- (v3);
\draw[edge, color=orange, opacity=.7] (v4) -- (v5) -- (w2) -- (v4);
\draw[edge, color=orange, opacity=.7] (v5) -- (v1) -- (w3) -- (v5);
\draw[edge, color=orange, opacity=.7] (w5) -- (z5);
\draw[edge, color=orange, opacity=.7] (w2) -- (z2);
\end{pgfonlayer}
\end{tikzpicture}
}
\hfill
\subfloat[\mbox{a minimal example}]{
\begin{tikzpicture}
\node[draw=none,minimum size=3em,regular polygon,rotate around={180:(0,0)},regular polygon sides=3] (a) {};
\node[draw=none,minimum size=7em,regular polygon,rotate around={180:(0,0)},regular polygon sides=3] (b) {};
\node[label=above:$\boldsymbol{x}$] at (a.corner 2){};
\node[dvertex,label=above:$\boldsymbol{y}$, label={[yshift=-2.1em, xshift=.5em]$\boldsymbol{e}$}] (w2) at (b.corner 2) {};
\foreach \x in {1,2,3}{
  \node[vertex] (v\x) at (a.corner \x) {};
  \node[vertex] (w\x) at (b.corner \x) {};
  \draw[edge] (v\x) -- (w\x);
}
\draw[edge] (v1) -- (v2) -- (v3) -- (v1);
\end{tikzpicture}
}
\hfill
\subfloat[the amplification]{
\begin{tikzpicture}
\node[draw=none,minimum size=3em,regular polygon,rotate around={180:(0,0)},regular polygon sides=3] (a) {};
\node[draw=none,minimum size=7em,regular polygon,rotate around={180:(0,0)},regular polygon sides=3] (b) {};

\foreach \x in {1,3}{
  \node[vertex] (v\x) at (a.corner \x) {};
  \node[vertex] (w\x) at (b.corner \x) {};
  \draw[edge] (v\x) -- (w\x);
}
\node[dvertex,label=above:$\boldsymbol{x}$] (v2) at (a.corner 2) {};
\node[label={[xshift=3.1em]the gadget}] at (b.corner 2){};
\draw[edge] (v1) -- (v2) -- (v3) -- (v1);
\begin{pgfonlayer}{background}
\draw[edge, color=orange, opacity=.7, line width=3em] (v2) -- (w2);
\end{pgfonlayer}
\end{tikzpicture}
}
\caption{The construction}
\label{fig:ampl}
\end{figure}

\begin{theorem}For every $\epsilon>\frac{1}{3}$, there exists a hypergraph $\mathcal{H}_\epsilon$ such that $\mathcal{H}_\epsilon$ has no isolated vertices, each hyperedge of $\mathcal{H}_\epsilon$ has size 2 or 3, $|V(\mathcal{H})|\geq (2-\epsilon)|E(\mathcal{H})|$, every proper subset of $E(\mathcal{H}_{\epsilon})$ has a 2-path cover, but $\mathcal{H}_{\epsilon}$ does not have a 2-path cover. 
\end{theorem}

\begin{proof}
Let $\epsilon>\frac{1}{3}$ and consider the gadget $\mathcal{G}$ shown in Figure~\ref{fig:ampl}.(a).  It is easy to verify that every 2-path cover of $\mathcal{G}$ must cover the vertex $x$.  
Next note that the hypergraph $\mathcal{H}$ shown in Figure~\ref{fig:ampl}.(b) is obviously not 2-path coverable, but every proper subset of $E(\mathcal{H})$ is
2-path coverable.  We have $\frac{|V(\mathcal{H})|}{|E(\mathcal{H})|}=\frac{6}{4}$.  However, we can increase this ratio via the amplification trick shown in 
Figure~\ref{fig:ampl}.(c).  

That is, let $e$ be a hyperedge of $\mathcal{H}$ of size 2.  Label the vertices of $e$ as $x$ and $y$, where $y$ has degree 1.  
Let $\mathcal{H}_1$ be the hypergraph obtained from $\mathcal{H}$ by deleting $y$ and then gluing $\mathcal{G}$ to  $\mathcal{H}-y$ along $x$.  Since every
2-path cover of $\mathcal{G}$ must use the vertex $x$, $\mathcal{H}_1$  does not have a 2-path cover.  On the other hand, since every proper subset
of $E(\mathcal{G})$ has a 2-path cover avoiding $x$, it follows that every proper subset of $E(\mathcal{H}_1)$ has a 2-path cover.  Note that this amplification 
trick increases the number of vertices of $\mathcal{H}$ by 10 and the number of edges of  $\mathcal{H}$ by 6.  
Moreover, we can repeat this amplification trick arbitrarily many times since $\mathcal{G}$ also has pendent edges of size 2.  So, choose $n$ such that $\frac{6+10n}{4+6n} \geq 2-\epsilon$ and
take  $\mathcal{H}_\epsilon$ to be the graph obtained from $\mathcal{H}$ by performing the amplification trick $n$ times.  
\end{proof}

\section{A cover game over bipartite graphs}
\label{sec:covergame}

As an application, we use the previous result to build a winning strategy for a game played on bipartite graphs.

\begin{definition}[Cover Game]  \label{def:cover-game}
The \dfn{Cover Game \CoverGame[h,k]{G,\mu}} is a game between two players, 
\Choose and \Cover, on a bipartite graph $G$. At each step $i$ of the game the players maintain an \matching[h,k] $F_i$ in $G$. 
At step $i+1$ \Choose can 
\begin{enumerate}
\item remove a connected component from $F_i$, or
\item if the number of connected components of $F_i$ is strictly less than $\mu$, pick a vertex and challenge \Cover to find another \matching[h,k] $F_{i+1}$ in $G$
such that
\begin{enumerate}
\item $F_{i+1}$ extends $F_i$.  That is, each connected component of $F_i$ is also a connected component of $F_{i+1}$;
\item $F_{i+1}$ covers the vertex picked by \Choose.
\end{enumerate}
\end{enumerate}
\Cover  loses the game \CoverGame[h,k]{G,\mu} if at some point she cannot answer a challenge by \Choose.  Otherwise, \Cover wins.  
\end{definition}

\begin{definition}[\expander{s,\delta}] Let $s$ be a positive integer and $\delta$ be a positive real number. A bipartite graph  $G$ with bipartition $(L,R)$ is an \dfn{\expander{\boldsymbol{s,\delta}}} if all subsets $X \subseteq L$ of size at most $s$ satisfy $|N_{G}(X)| \geq \delta|X|$.
\end{definition}

The next theorem shows that \Cover has a winning strategy for the game \CoverGame{G, \mu} for expander graphs $G$ with appropriately chosen parameters.

\begin{theorem}\label{thm:covergame}
Let $G$ be a bipartite graph with bipartition $(L,R)$, $s,d$ be two integers, and $\epsilon <\frac{1}{23}$ be a real number.  If each vertex in $L$ has degree $3$, each vertex in $R$ has degree at most $d$, and $G$ is an \expander{s,2-\frac{\epsilon}{2}}, then \Cover wins the cover game \CoverGame{G,\mu} with 
$\mu=\frac{\epsilon s}{72d+2\epsilon}-d$.
\end{theorem}

The proof of this result is similar to constructions that can be found for example in \cite{Ben-SassonG03,Atserias04, BonacinaGT14}. The proof we give 
is modeled following \cite{BonacinaGT14}. The main difference is the use of our variant of Hall's Theorem (Theorem \ref{thm:Hall1}), suitable for bipartite graphs of left-degree at most $3$.

For the rest of this section, fix a bipartite graph $G$ with bipartition $(L,R)$, integers $s,d$ and a real number $\epsilon < \frac{1}{23}$ such that $G$ is an $(s,2-\frac{\epsilon}{2})$-expander, each vertex in $L$ has degree $3$ and each vertex in $R$ has degree at most $d$.

Given $A\subseteq L$ and $B\subseteq R$, we denote with $G_{A,B}$ the subgraph of $G$ induced by $(L\cup R) \setminus (A\cup B)$.

\begin{definition}[\robust]
Given two sets $A \subseteq L$ and $B \subseteq R$, we say that the pair $(A,B)$ is \dfn{\robust}, if for every $C \subseteq L \setminus A$ with
$|A| + |C| \leq s$, there exists a \matching $F$ in $G_{A,B}$ covering $C$.
\end{definition}

\begin{lemma} \label{lem:empty}
The pair $(\emptyset,\emptyset)$ is \robust.
\end{lemma}

\begin{proof}
This follows from Theorem \ref{thm:Hall1} since $G$ is an $(s,2-\frac{\epsilon}{2})$-expander.
\end{proof}

\begin{lemma}\label{lem:smallC}
Let $A\subseteq L$ and $B\subseteq R$ be such that the pair $(A,B)$ is not \robust.  Then there exists a set $C\subseteq L\setminus A$ such that no \matching in $G_{A,B}$ covers $C$ and $|C|< \frac{2}{\epsilon}|B|$.
\end{lemma}
\begin{proof}
Take $C\subseteq L\setminus A$ of minimal size such that no \matching in $G_{A,B}$ covers $C$. By minimality of $C$ and Theorem \ref{thm:Hall2} it follows that 
$$
|N_{G_{A,B}}(C)|< (2-\epsilon)|C|.
$$
But, by hypothesis on $G$, $(2-\frac{\epsilon}{2})|C|\leq |N_G(C)|$. Therefore,
$$
(2-\frac{\epsilon}{2})|C|\leq |N_G( C)|\leq |N_{G_{A,B}}(C)|+ |B|<(2-\epsilon)|C|+|B|.
$$
Hence $|C|<\frac{2}{\epsilon} |B|$, as required.
\end{proof}

\begin{lemma}[component removal]\label{lem:comp-removal}
Let $A\subseteq L$ and $B\subseteq R$ be such that the pair $(A,B)$ is \robust and $\frac{2}{\epsilon}|B|+|A|\leq s$. 
Then for each \matching $F$ over the vertices $A\cup B$,
$(A\setminus \ell(F),B\setminus r(F))$ is \robust.
\end{lemma}

\begin{proof}
Let $A'=A\setminus \ell(F)$ and $B'=B\setminus r(F)$ and suppose, by contradiction, that $(A',B')$ is not \robust. 
By Lemma \ref{lem:smallC}, it is sufficient to prove that for each set $C\subseteq L\setminus A'$ with $|C|<\frac{2}{\epsilon}|B'|$, there is a \matching in $G_{A',B'}$ covering $C$. Let $C'=C\cap \ell(F)$ and $C''=C\setminus C'$. By construction, $F$ is a \matching such that $\ell(F)\subseteq A$, $r(F)\subseteq B$ and $F$ covers $C'$. Moreover, we have that
 $$
 |C''|+|A|\leq|C|+|A|<\frac{2}{\epsilon}|B'|+|A|<\frac{2}{\epsilon}|B|+|A|\stackrel{(\star)}{\leq} s,
 $$
 where the inequality $(\star)$ is by hypothesis. Hence there exists a \matching $F''$ of $C''$ in $G_{A,B}$, and so $F \cup F''$ is a \matching covering $C$ in $G_{A',B'}$.
\end{proof}

\begin{lemma}[left vertex-covering] \label{lem:left-vertex-cover}
Let $A\subseteq L$ and $B\subseteq R$ be such that the pair $(A,B)$ 
is \robust and
$\frac{24d}{\epsilon}(|B|+3) + |A| + 1 \le s$.
Then for each vertex $v$ in $L\setminus A$, there is a \matching $F$ in $G_{A,B}$ covering $v$ and such that $(A \cup \ell(F), B \cup r(F))$ is \robust.
\end{lemma}

\begin{proof}
Fix $v\in L\setminus A$ and let $\Pi$ be the set of all \matchings $F$ in $G_{A,B}$, covering $v$ and such that $F$ is connected.

Since $|A|  + 1 \le s$ and $(A,B)$ is \robust, we know that
$\Pi$ is non-empty. 
For every $F\in \Pi$, let $(A_F,B_F)$ be the pair $(A \cup \ell(F), B \cup r(F))$, and suppose for a contradiction that for every $F \in \Pi$, $(A_F, B_F)$ is not \robust. 
By Lemma \ref{lem:smallC}, for every $F \in \Pi$ there is a set
$C_F \subseteq L \setminus A_F$
with $|C_F| < \frac{2}{\epsilon}|B_F|$ and such that there is no \matching 
of $C_F$ in $G_{A_F,B_F}$.

Let $C = \bigcup_{F \in \Pi} C_F$.
Then $|C| < \frac{24d}{\epsilon}(|B|+3)$, since $|\Pi| \leq 3+ 3\cdot 2\cdot (d-1) \cdot 2\leq 12 d$ and $|r(F)|\leq 3$. Hence,
by our assumption about the sizes of $|A|$ and $|B|$, we have that $|C \cup \{v \}| + |A|\le s$.
Furthermore, $C \cup \{ v \} \subseteq L \setminus A$, so by the fact that $(A,B)$ is \robust, there is a \matching $F'$ covering $C \cup \{ v \}$ in $G_{A,B}$.

There must be some $F \in \Pi$ such that $F$ is a connected component of $F'$. Let $F''$ be $F'$
with the component $F$ removed. Then $F''$ 
is a \matching in $G_{A_F,B_F}$ and $F''$ covers $C_F$, contradicting the choice of $C_F$. 
\end{proof}

\begin{lemma}[right vertex-covering] \label{lem:right-vertex-cover}
Let $A\subseteq L$ and $B\subseteq R$ be such that the pair $(A,B)$ 
is \robust and
$\frac{24d}{\epsilon}(|B|+3d) + |A| + 2d \le s$.
Then for each vertex $v$ in $R\setminus B$, there is a \matching $F$ in $G_{A,B}$ covering $v$ and such that $(A \cup \ell(F), B \cup r(F))$ is \robust.
\end{lemma}

\begin{proof}
Fix $v \in R\setminus B$ and let $D$ be $N_G(v)\setminus A$. By hypothesis $|D|\leq d$ and, by the cardinality constraints on $A$ and $B$, we can apply Lemma \ref{lem:left-vertex-cover} $|D|$ times obtaining a \matching $F$ in $G_{A,B}$ covering $D$ and such that $(A \cup \ell(F), B \cup r(F))$ is \robust. Now, since $N_G(v)\subseteq A\cup \ell(F)$, we have that $(A\cup \ell(F),B\cup r(F)\cup \{v\})$ is \robust. Either $v$ is covered by $F$, or it is possible to add $\{v\}$ as a new connected component to $F$ while still maintaining the property of being a \matching in $G_{A,B}$.
\end{proof}

\begin{proof}[Proof of Theorem \ref{thm:covergame}]
We now use the previous lemmas to describe a winning strategy \LL for \Cover to win \CoverGame{G,\mu}. Take \LL to be the set of all \matchings $F$ in $G$ such that 
\begin{enumerate}
\item $(\ell(F),r(F))$ is \robust, and
\item $\frac{2}{\epsilon}|r(F)|+|\ell(F)|\leq s$.
\end{enumerate}

This family is non-empty by Lemma \ref{lem:empty} and is closed under removing connected components by Lemma \ref{lem:comp-removal}. Suppose now that at step $i+1$ of the game \Choose picks a vertex $v$ in $G$ and that $F_i$ has strictly  less than 
$\mu=\frac{\epsilon s}{72d+2\epsilon}-d$ components.  
 Then, $(\ell(F_i),r(F_i))$ satisfies the hypotheses of Lemma \ref{lem:left-vertex-cover} and Lemma \ref{lem:right-vertex-cover}:
\begin{align*}\label{eq:boh}
\frac{24d}{\epsilon}(|r(F_i)|+3d) + |\ell(F_i)| + 2d 
\leq 
\frac{24d}{\epsilon}(3\mu+3d) + 2\mu + 2d 
=(\mu +d) (\frac{72d}{\epsilon}+2)
=
\\
=
(\frac{\epsilon s}{72d+2\epsilon})(\frac{72d}{\epsilon}+2)= s.
\end{align*}
Hence by  Lemma \ref{lem:left-vertex-cover} and Lemma \ref{lem:right-vertex-cover}, there exists a \matching $F_{i+1}$ extending $F_i$ and covering $v$ such that 
$(\ell(F_{i+1}),r(F_{i+1}))$ is \robust.
From the previous chain of inequalities, it follows easily that the pair $(\ell(F_{i+1}),r(F_{i+1}))$ satisfies the cardinality condition $\frac{2}{\epsilon}|r(F_{i+1})|+|\ell(F_{i+1})|\leq s$.
\end{proof}

\section{Space lower bounds for random 3-CNFs in \RES and \PCR}
\label{sec:lowerbounds}

To prove a space lower bound for random $3$-CNFs, we use the techniques developed in \cite{bg15,BonacinaGT14}. 
Using the Cover Game on \matchings, we provide families of  assignments that define {\em \WSs{k}} (for the monomial space in \PCR) \cite{bg15} and {\em $r$-free families} (for the total space in Resolution) \cite{BonacinaGT14}.  
We start with some preliminaries, and then prove that {\em $k$-winning strategies} alone are sufficient. Finally, we prove the main results, which together with the main Theorem  in \cite{BonacinaGT14}  will imply the two results on space bounds. 

\subsection{Flippable assignments from \matchings[\boldsymbol{2,4}]} 

Let $A$ be a family of partial assignments, and let $\dom(A)$ be the union of the domains of the assignments in $A$. 
We say that a set of partial assignments $A$  is \dfn{flippable}  if and only if  for all $x \in \dom(A)$ there exist $\alpha,\beta \in A$ such that $\alpha(x)=1-\beta(x)$.
Two families of partial assignments $A$ and $A'$ are \dfn{domain-disjoint} if $\dom(\alpha)$ and $\dom(\alpha')$ are disjoint for all $\alpha\in A$ and $\alpha'\in A'$.
Given non-empty and pairwise domain-disjoint sets of assignments\footnote{We always suppose that the partial assignments are respecting the intended meaning of the variables in $\overline X$.  That is, if $x\in \dom(\alpha)$, then $\alpha(\bar x)=1-\alpha(x)$; hence a variable $x$ is in $\dom(H_i)$ if and only if $\bar x$ is in $\dom(H_i)$.} $\xvec{H}{1}{t}$, the \dfn{product-family} $\HH=H_1\otimes \ldots \otimes H_t$ is the following set of assignments
$$
\HH=H_1\otimes \ldots \otimes H_t= \{\alpha_1\cup\ldots \cup \alpha_t \st \alpha_i\in H_i\},
$$
or, if $t=0$, $\HH=\{\Star\}$, where $\Star$ is the partial assignment of the empty domain.
Note $\dom(\HH)=\bigcup_{i}\dom(H_i)$.  We call the $H_i$ the \dfn{factors} of \HH.
For a product-family $\HH= H_1\otimes \ldots \otimes H_t$, the \dfn{rank} of \HH, denoted $\|\HH\|$, is the number of factors of $\HH$ different from $\{\Star\}$. We do not count $\{\Star\}$ in the rank since $\HH\otimes \{\Star\}=\HH$. 
Given two product-families \HH and $\HH'$, we write $\HH'\sqsubseteq \HH$ if and only if each factor of $\HH'$ different from $\{\Star\}$ is also a factor of \HH. In particular, $\{\Star\}\sqsubseteq \HH$ for every \HH.

A family of \FPFs is called a \dfn{strategy} and denoted by \LL.  
We now present a definition of suitable families of flippable products: the {\em \WSs{k}} \cite{bg15}\footnote{This definition correspond to a particular case of the definition of {\em \WS{k}} in the  \cite{bg15}. In \cite{bg15} the definition of \WS{k} depend on the choice of some particular proper ideal $I$ in $\mathbb{F}[X,\overline{X}]$. Here we just state the simplified case where $I=\{0\}$.}.

\begin{definition}[\WS{k}\cite{bg15}]
\label{def:kex}
Let $P$ be a set of polynomials in the ring $\mathbb{F}[X,\overline X]$. A non-empty strategy \LL is a \dfn{\WS{\boldsymbol{k}}} if and only if for every $\HH \in  \LL$ the following conditions hold:
 \begin{description}
 \item[(restriction)] for each $\HH' \sqsubseteq \HH$, $\HH' \in \LL$;
 \item[(extension)] if  $\|\HH\|<k$, then for each $p \in P$  there exists a \FPF $\HH'\in \LL$ such that $\HH'\sqsupseteq \HH$ and $\HH' \models p$.
 \end{description}
\end{definition}
Notice that, by the restriction property, $\{ \Star\}$ is in every \WS{k}.

\begin{theorem}[\cite{bg15}] \label{thm:lowerbound}
Let $P$ be a contradictory set of polynomials in $\mathbb{F}[X,\overline X]$ and $k\geq 1$ an integer.  If there exists a non-empty \WS{k}  \LL for $P$, then
for every semantical \PCR refutation $\Pi$ of $P$, $\Space(\Pi) \geq k/4$.
\end{theorem}

\subsection{Total space lower bounds via $\boldsymbol{k}$-winning strategies}
In this section we recall some definitions and theorems from \cite{BonacinaGT14} and we show that in order to prove total space lower bounds in \RES it suffices to show the existence of \WSs{k}.

A \dfn{piecewise (p.w.) assignment} $\alpha$ of a set of variables $X$ is a set of non-empty partial assignments to $X$ with pairwise disjoint domains.
We will sometimes call the elements of $\alpha$ the \emph{pieces} of $\alpha$. 
A piecewise assignment gives rise to  a partial assignment $\bigcup \alpha$ to $X$ together
with a partition of the domain of $\bigcup \alpha$. 

For piecewise assignments $\alpha, \beta$ we will write $\alpha \sqsubseteq \beta$ to mean 
that every piece of $\alpha$ appears in $\beta$. We will write $\| \alpha \|$ to mean the number
of pieces in~$\alpha$. Note that these are formally exactly the same as $\alpha \subseteq \beta$
and $|\alpha|$, using the definition of $\alpha$ and $\beta$ as sets.

\begin{definition}[$r$-free \cite{BonacinaGT14}]\label{def:r-free}
A non-empty family $\FF$ of p.w. assignments is 
\emph{$r$-free for a CNF $\phi$} if it has the following properties.
\begin{description}
\item[(Consistency)]
No $\alpha \in \FF$ falsifies any clause from $\phi$.
\item[(Retraction)]
If $\alpha \in \FF$, $\beta$ is a p.w. assignment and $\beta \sqsubseteq \alpha$, then $\beta \in \FF$. 
\item[(Extension)]
If $\alpha \in \FF$ and $\|\alpha\| < r$, then for every variable 
$x \notin \dom (\alpha)$, there exist $\beta_0, \beta_1 \in \FF$ 
with $\alpha \sqsubseteq \beta_0, \beta_1$ such that 
$\beta_0(x)=0$ and $\beta_1(x)=1$.
\end{description}
\end{definition}

\begin{theorem}[\cite{BonacinaGT14}] \label{thm:main-free}
Let $\phi$ be an unsatisfiable CNF formula.
If there is a family of p.w. assignments which is $r$-free for $\phi$,
then any resolution refutation of $\phi$ must pass through a memory
configuration containing at least $r/2$ clauses each of width at least $r/2$.
In particular, the refutation requires total space at least $r^2/4$.
\end{theorem}
 
The next lemma establishes a connection between Definition \ref{def:kex} and Definition \ref{def:r-free}. 

\begin{lemma}
\label{lem:winning-free}
Let $\phi$ be an unsatisfiable CNF and let $P_\phi = tr(\phi)$ be the standard polynomial encoding of $\phi$.  If there exists  a \WS{k} for $P_\phi$, then there exists a $(k-1)$-free  family for $\phi$.
\end{lemma} 

\begin{proof}
Let \LL be the \WS{k}.
Define the  $(k-1)$-free  family \FF as follows:
$\alpha \in \FF$ if and only if  there exists $H_1\otimes \ldots \otimes H_t \in \LL$ such that $\alpha = \alpha_1\cup\ldots \cup \alpha_t$
with $\alpha_i\in H_i$ and $t \leq k-1$. 
The p.w. structure of $\alpha$ is inherited from the domain-disjointness of $H_1\otimes \ldots \otimes H_t$; in particular, $\|\alpha\|=\|H_1\otimes \ldots \otimes H_t\|$.
The {\em retraction property} of \FF is immediate from the corresponding property of \LL.

To prove the {\em consistency property} of \FF assume, by contradiction, that there is an $\alpha \in \FF$ such that $\alpha$ falsifies some clause $C \in \phi$. 
Since $||\alpha||\le k-1<k$,  there exists $\HH= H_1\otimes \ldots \otimes H_t \in \LL$ such that $\alpha\in \HH$ and $\|\alpha\|=\|\HH\|$. By the extension property of \LL, there is an $\HH'\sqsupseteq \HH$ such that $\HH'\models tr(C)$.  In particular there exists some partial assignment $\beta\supseteq \alpha$ such that $\beta\models tr(C)$.  By construction, for every assignment $\gamma$, $\gamma\models tr(C)$ if and only if $\gamma\models C$.  Thus $\beta\models C$, which is impossible since $\alpha$ falsifies $C$.

For  the {\em extension property} let $\alpha \in \FF$, with $||\alpha||<k-1$ and let $x$ be a variable of $\phi$ not in $\dom(\alpha)$. By construction, there exists some $\HH\in \LL$ such that $\alpha\in \HH$, $\|\alpha\|=\|\HH\|$ and $\dom(\alpha)=\dom(\HH)$. By the extension property of \FF there exists some flippable $\HH'\in \LL$ such that $\HH'\sqsupseteq \HH$ and $\HH'\models x^2-x$. By taking restrictions in \LL we can suppose that $\|\HH'\|=\|\HH\|+1$.
Hence there exist $\beta_0,\beta_1\in \FF$ extending $\alpha$, setting $x$ respectively to $0$ and $1$ and such that  $||\beta_0||=||\beta_1||=\|\alpha\|+1\le k-1$. 
\end{proof}

To avoid confusion, we want to emphasize that this reduction from \WSs{k} to $(k-1)$-free families does not hold for the full general definition of \WSs{k} as in \cite{bg15}. 

\subsection{Space lower bounds via $\boldsymbol{\mathsf{CoverGames}}$ on  $\boldsymbol{(2,4)}$-matchings}

\begin{table}[h!]
\caption{From \matchings to flippable assignments}\label{table:flippable}
\footnotesize
\centering
\begin{tabular}{c|c|c|c}
\begin{tikzpicture}
\node[Variable] (v1) {};
\node[Variable,  below of = v1] (v2) {};
\node[Variable, Grey, below of = v2, , label= right:$x$] (v3) {};
\node[Variable, below of = v3] (v4) {};
\node[Variable, below of = v4] (v5) {};
\node[node distance=3em, left of = v3] (v6) {};
\end{tikzpicture}
&
\begin{tikzpicture}
\node[Variable] (v1) {};
\node[Variable, Grey, below of = v1, label= right:$x$] (v2) {};
\node[Variable, below of = v2] (v3) {};
\node[Variable, Grey, below of = v3, label= right:$y$] (v4) {};
\node[Variable, below of = v4] (v5) {};
\node[Clause, left of = v3, label= left:$x\vee y$] (v6) {};
\draw[Edge] (v2) -- (v6) -- (v4);
\end{tikzpicture}
&
\begin{tikzpicture}
\node[Variable, Grey, label=right:$x$] (v1) {};
\node[Variable, below of = v1] (v2) {};
\node[Variable, Grey, below of = v2, label=right:$y$] (v3) {};
\node[Clause, left of= v2,label=left:$x\vee y$] (v4) {};
\node[Variable,  below of = v3] (v5) {};
\node[Variable, Grey, below of = v5, label=right:$z$] (v6) {};
\node[Clause,  left of= v5,label=left:$y\vee z$] (v7) {};
\draw[Edge] (v1) -- (v4) -- (v3) -- (v7) -- (v6);
\end{tikzpicture} 
&
\begin{tikzpicture}
\node[Variable, Grey, label=right:$x$] (v1) {};
\node[Variable, below of = v1] (v2) {};
\node[Variable, Grey, below of = v2, label=right:$y$] (v3) {};
\node[Clause, left of= v2,label=left:$x\vee y$] (v4) {};
\node[Variable,  below of = v3] (v5) {};
\node[Variable, Grey, below of = v5, label=right:$z$] (v6) {};
\node[Clause, left of= v5,label=left:$\neg y\vee z$] (v7) {};
\draw[Edge] (v1) -- (v4) -- (v3) -- (v7) -- (v6);
\end{tikzpicture}
\\
$\begin{array}{c}
x\mapsto 0 \\
 x\mapsto 1
\end{array}
$ 
&
$\begin{array}{c}
(x,y)\mapsto (0,1) \\
 (x,y)\mapsto (1,0)
\end{array}
$
&
$\begin{array}{c}
(x,y,z)\mapsto (0,1,0) \\
 (x,y,z)\mapsto (1,0,1)
\end{array}
$
&
$\begin{array}{c}
(x,y,z)\mapsto (0,1,1) \\
 (x,y,z)\mapsto (1,0,0)
\end{array}
$
\end{tabular}
\end{table}

\begin{lemma}
\label{lem:match-flip}
Let $\phi$ be a CNF and $G_\phi$ be the adjacency graph of $\phi$. For every \matching $F$ in $G_\phi$, there exists a flippable product-family of assignments $H_F$ such that $H_F\models \ell(F)$, $\dom(H_F)=r(F)$, and $\|H_F\|$ is the number of connected components of $F$.
\end{lemma}

\begin{proof}
We prove the result by induction on the number of connected components of $F$.
If $F$ is the union of two disjoint \matchings $F',F''$ then by hypothesis $H_{F'}\models \ell(F')$, $\dom(H_{F'})=r(F')$ and $\|H_{F'}\|$ is the number of connected components of $F'$. And analogously for $F''$. Then, since $r(F')$ and $r(F'')$ are disjoint, $H_F=H_{F'}\otimes H_{F''}$ is well-defined. We immediately see that $H_{F}\models \ell(F)$, $\dom(H_{F})=r(F)$ and $\|H_{F}\|$ is the number of connected components of $F$.

It remains to prove the case when the \matching $F$ is just one connected component. It is easy to see that all the possibilities can be reduced to those in Table \ref{table:flippable}.
\end{proof}

\begin{theorem}
\label{thm:cover-to-strategy}
Let $\phi$ be an unsatisfiable $3$-CNF and $G_\phi$ its adjacency graph. If \Cover wins the cover game \CoverGame{G_\phi,\mu},  then there is a \WS{\mu} \LL for 
$tr(\phi)$.
\end{theorem}

\begin{proof}
It is straightforward to see that a winning strategy for \Cover in \CoverGame{G_\phi,\mu} defines, using Lemma \ref{lem:match-flip},  a family \LL of flippable product-families such that for all $\HH\in \LL$
 \begin{enumerate}
 \item for each $\HH' \sqsubseteq \HH$, $\HH' \in \LL$;
 \item if  $\|\HH\|<\mu$, then:  (a) for each $C \in \phi$, there exists a \FPF $\HH'\in \LL$ such that $\HH'\models C$ and $\HH'\sqsupseteq \HH$; and (b) for each variable $x\not \in \dom(\HH)$, there exists a flippable family $\HH'\in \LL$ such that $\HH'\sqsupseteq \HH$.
 \end{enumerate}

We claim that \LL is a \WS{\mu}. The {\em size property} and the {\em restriction property} are immediate. For the {\em extension property} we use the properties in (2) above: if we have to extend to something in \LL that satisfies a boolean axiom we use property 2.(b), otherwise for all other polynomials in $tr(\phi)$ we use property 2.(a). 
\end{proof}

\begin{corollary}
\label{cor:covgame-space}
Let $\phi$ be an unsatisfiable $3$-CNF and $G_\phi$ its adjacency graph. If \Cover wins the cover game \CoverGame{G_\phi,\mu},  then for every semantical \PCR refutation $\Pi$ of $tr(\phi)$, $\Space(\Pi) \geq \mu/4$.
Moreover, every \RES refutation of $\phi$ must pass through a memory
configuration containing at least $\mu/2$ clauses each of width at least $\mu/2$. In particular, the refutation requires total space at least $\mu^2/4$.
\end{corollary}

\begin{proof}
The result follows from Theorem \ref{thm:cover-to-strategy}, Theorem \ref{thm:lowerbound}, Lemma \ref{lem:winning-free}, and Theorem \ref{thm:main-free}.
\end{proof}

Let $n,\Delta\in \mathbb{N}$ and let $X=\{x_1, \ldots , x_n\}$ be a set of $n$ variables.   The probability distribution $\R(n,\Delta,3)$ is obtained by the following experiment: choose independently uniformly at random $\Delta n$ clauses from the set of all possible clauses with $3$ literals over $X$.  
It is well-known that when $\Delta$  exceeds a certain constant $\theta_3$, $\phi$ is almost surely unsatisfiable. 
Hence we always consider $\phi\sim \R(n,\Delta,3)$, where $\Delta$ is a constant bigger than $\theta_3$, which implies that $\phi$ is unsatisfiable with high probability. 

\begin{corollary} \label{cor:rand}
Let $\Delta > \theta_3$
and  $\phi\sim \R(n,\Delta,3)$. Then with high probability, for every semantical \PCR refutation $\Pi$ of $tr(\phi)$, $\Space(\Pi) \geq \Omega(n/\log n)$.
 Moreover, every \RES refutation of $\phi$ must pass through a memory
configuration containing $\Omega(n/\log n)$ clauses each of width $\Omega(n/\log n)$. In particular, each refutation of $\phi$ requires total space  $\Omega(n^2/\log^2 n)$.
\end{corollary}
\begin{proof}
Let $G_\phi$ be the adjacency graph of $\phi$. By Chernoff's bound it follows that with high probability the right-degree of $G_\phi$ is at most $\log n$ and it is well known that $G_\phi$ is a \expander{\gamma n, 2-\delta}, for every $\delta>0$ \cite{ChvatalS88,BeameP96,Ben-SassonW01,Ben-SassonG03}. 
Hence in particular for $0<\delta < \frac{1}{23}$,  we can apply Theorem \ref{thm:covergame}: \Cover wins the cover game $\CoverGame{G_\phi,\mu}$ for $\mu = \Omega(n/ \log n)$. The result then follows using Corollary \ref{cor:covgame-space}.
\end{proof}

For every unsatisfiable CNF in $n$ variables, there is a trivial $O(n)$ upper bound on monomial space in \PCR and a trivial $O(n^2)$ upper bound on total space in \RES. It is an open question whether our lower bounds are tight.

\bibliographystyle{plain}

\bibliography{bibliography}



\end{document}